\documentclass[letterpaper, 10 pt, conference]{ieeeconf}
\IEEEoverridecommandlockouts

\usepackage[colorlinks]{hyperref}
\hypersetup{
    colorlinks,
    linkcolor=black,
    citecolor=black,
    filecolor=magenta,
    urlcolor=cyan,
}

\usepackage{resizegather}

\usepackage{cite}
\usepackage{amsmath,amssymb,amsfonts,mathrsfs}
\usepackage{graphicx}
\usepackage{textcomp}
\usepackage{xcolor}
\usepackage{tcolorbox}
\usepackage{textcomp}
\usepackage{tablefootnote}
\usepackage{threeparttable}
\usepackage{caption, subcaption}
\usepackage{bbm}
\usepackage{dsfont}
\usepackage{tikz}
\usepackage{graphicx}
\usepackage{tkz-euclide}
\usepackage{algorithm}
\usepackage{algpseudocode}
\usetikzlibrary{positioning}
\usetikzlibrary{shapes}
\usetikzlibrary{shapes.misc}
\usetikzlibrary{shapes.geometric}
\usetikzlibrary{plotmarks}
\usetikzlibrary{intersections}
\usetikzlibrary{calc}
\usetikzlibrary{fit}
\usetikzlibrary{patterns,tikzmark}
\usetikzlibrary{matrix,decorations.pathreplacing,calc}

\tikzset{cross/.style={cross out, draw, 
         minimum size=2*(#1-\pgflinewidth), 
         inner sep=0pt, outer sep=0pt}}
\newcommand{\state}[0]{x}

\newcommand{\norm}[1]{\left\lVert#1\right\rVert}

\newcommand{\vertiii}[1]{{\left\vert\kern-0.25ex\left\vert\kern-0.25ex\left\vert #1 
    \right\vert\kern-0.25ex\right\vert\kern-0.25ex\right\vert}}

\newtheorem{theorem}{Theorem}
\newtheorem{corollary}{Corollary}[theorem]
\newtheorem{lemma}{Lemma}

\newtheorem{remark}{Remark}

\newtheorem{definition}{Definition}

\makeatletter
\renewcommand{\fps@figure}{htp}
\renewcommand{\fps@table}{htp}
\makeatother
\def\BibTeX{{\rm B\kern-.05em{\sc i\kern-.025em b}\kern-.08em
    T\kern-.1667em\lower.7ex\hbox{E}\kern-.125emX}}

\begin{document}

\title{Prescribed Time Higher Order Control Lyapunov-Barrier Functions}
% Some Random Title can be changed later

\author{Janani S K$^{1}$, Manan Tayal$^{2}$, and  Shishir Kolathaya$^{2}$% % <-this % stops a space
\thanks{
}
\thanks{$^{1}$Department of Engineering Design, Indian Institute of Technology, Madras.
{\tt\scriptsize ed22b017@smail.iitm.ac.in}
.
}% <-this % stops a space
\thanks{$^{2}$ Center for Cyber-Physical Systems, Indian Institute of Science (IISc), Bengaluru.
{\tt\scriptsize \{manantayal, shishirk\}@iisc.ac.in}
.
}%
}

\newcommand{\mcom}[1]{{\color{red}{{[MT: #1]}}}}
\newcommand{\jcom}[1]{{\color{blue}{{[J: #1]}}}}
\maketitle
\begin{abstract}
We propose a framework based on \emph{Repeated-Root Prescribed Time Higher-Order Control Lyapunov- Barrier Functions} (RR-Pxt-HOCL-BF)  that guarantees convergence to a safe set/ goal set within user-specified time. Unlike conventional Higher-Order Control Barrier Functions, which does not provide any guarantee on reaching the safe set if the system is initially not in safe set, the proposed method enforces convergence within a prescribed time horizon. This is achieved through a structured differential constraint derived from repeated roots of the characteristic polynomial, leading to closed-form polynomial solutions that ensure convergence at the desired time. Analytical conditions for forward invariance and prescribed-time reachability are derived, with explicit formulations for second-order systems. The approach is validated on an acceleration-controlled unicycle and an inverted pendulum and compared with existing HOCBF methods. Simulation results demonstrate reliable convergence within the desired time horizon, even in scenarios where conventional approaches fail.
\end{abstract}

\section{Introduction}
\label{section: Introduction}
In recent years, advancements in safety-critical controllers has enabled robots to operate in challenging environments. Researchers have developed various control methods for ensuring the safety of such systems which includes artificial potential fields  \cite{8022685, Singletary2021ComparativeAO}%
, reference governor \cite{6859176}, reachability analysis \cite{8263977 % RA-UAV, https://doi.org/10.48550/arxiv.2106.13176%
}, and nonlinear model predictive control \cite{8442967 %yu_mpc_aut_ground_vehicle%
}.
To establish formal safety guarantees, such as collision avoidance with obstacles, it is essential to employ a safety-critical controller that prioritizes safety over tracking. One such approach is based on Control Barrier Functions (CBFs) \cite{7040372, Ames_2017}, which define a secure safe set through inequality constraints and formulate it as a quadratic programming (QP) problem to ensure forward invariance of these sets over time.

A significant advantage of using CBF-based quadratic programs over other techniques is that they work efficiently on practical real-time applications in complex dynamic environments \cite{Singletary2021ComparativeAO, https://doi.org/10.48550/arxiv.2106.13176}; that is, optimal-safe control inputs can be computed at a very high frequency. It can be applied as a fast safety filter over existing path planning controllers\cite{9682803}, making it highly applicable to real-world autonomous systems, including Unmanned Ground Vehicles (UGV) and Unmanned Aerial Vehicles (UAV). Some works show collision avoidance with a static obstacle for point mass models \cite{7040372  %https://doi.org/10.48550/arxiv.2103.12382, 9565037%
}, for ground vehicles \cite{XIAO2021109592 %https://doi.org/10.48550/arxiv.2106.05341, 9482979, goswami2024collision, 9029446% 
}, and for aerial vehicles \cite{tayal2024control, 9112342}. However, these works do not guarantee safety in prescribed time when the system is not initally inside the safe set.

All these works assume that the system is initially in safe set and the safe controllers guarantee that system stays within the safe set, but these guarantees don't hold when the system is in unsafe region initially. In this work, the aim is to develop a safe controller that reaches safe set within prescribed time $T$ and remain within the safe set there after. Finite time convergence ensures that system reaches safe set in finite time which depends on the initial state, fixed time convergence ensures that the system reaches the safe set in finite time $T$ irrespective of the initial state. Prescribed time convergence ensures that the system reaches the safe set within $T$ that is prescribed by the user. 

Method to guarantee prescribed time convergence in first order CLBFs is proposed in \cite{9482751}. The closest work to prescribed time convergence in Higher order CBF is  \cite{xiao2021highordercontrollyapunovbarrier}, but this can only be used in cases $x(T_i)$ -intermediate states are known, thus they apply it to Temporal Logic specifications (labeled as STL-HOCBF). Even though the proposed method is compared with STL-HOCBF, it is important to note that the requirements and uses of these are different.

The main contributions are the following:
\begin{itemize}
    \item Introduce Repeated- Roots based higher order prescribed time converging CBFs. 
    \item Provide closed-form solutions for second order CBF/CLF.
\end{itemize}

Section \ref{section: Background} provides the preliminaries, the notion of Safe controller, CBFs, CBF-QP formulation, Higher order CBFs and first order fixed time convergence. The proposed method, repeated-roots based RR-Pxt-HOCBFs is introduced in section \ref{section: Method}, first the general \(n^{th}\) order CBF case and then explicitly shown for second order case . Section \ref{section: Simulation} contains the simulation results along with comparison with the existing methods and section \ref{section: Conclusions} includes the future scope and the limitations of the method proposed. Note that all the analysis can be extended to CLFs as well, but in this paper only the CBFs are considered. 
\section{Preliminaries}
\label{section: Background}
In this section we formally introduce Control- Barrier Functions, Higher order CBFs, First order finite time CBFs which will be used further.
\subsection{Safe controller}
Consider a continuous time control system with state $x(t) \in \mathcal{X}\subseteq\mathbb{R}^{n}$ and input $u(t) \in \mathbb{U} \subseteq \mathbb{R}^{m}$ at time $t\geq0$. The state dynamics is described by the following equation:
\begin{align*}\label{eq:state-ode}
    \Dot{x(t)}=f(x(t))+g(x(t))u(t)
\end{align*}
where $f: \mathbb{R}^n \rightarrow\mathbb{R}^n$ and $g: \mathbb{R}^n \rightarrow \mathbb{R}^{n\times m}$ are locally Lipschitz.
\newline
A set $\mathcal{C}$ is defined as \textit{super-level set} of a continuously differentiable function $h: \mathcal{X} \rightarrow \mathbb{R}$.
\begin{align*}
    \mathcal{C}=\{x: x \in \mathcal{X}, h(x) \geq 0\} \\
    \mathcal{X-C}= \{x: x \in \mathcal{X}, h(x) <0\}\\
    \text{Int}(\mathcal{C})=\{x: x \in \mathcal{X}, h(x) >0\}\\
    \partial\mathcal{C}=\{x: x \in \mathcal{X}, h(x)=0\}
\end{align*}
Let $\mu: \mathbb{R}^n \rightarrow\mathbb{R}^m$ be a mapping from sequence of states $x(t) $ to control input $u(t)$. Safety of the controller is defined as :
\begin{definition}[Safety]
    \label{def:positive-invariance}
    A set $\mathcal{C} \subseteq \mathcal{X} \subseteq \mathbb{R}^{n}$ is positive invariant under dynamics \eqref{eq:state-ode} and control policy $\mu$ if $x(0)\in \mathcal{C}$ and $u(t) = \mu(x(t)) \ \forall t \geq 0$ imply that $x(t) \in \mathcal{C}$ for all $t \geq 0$. 
    If $\mathcal{C}$ is positive invariant, then the system satisfies the safety constraint with respect to $\mathcal{C}$.
\end{definition}
\subsection{Control barrier function (CBFs)}
Control barrier functions are used to render the set $\mathcal{C}$ which is a super-level set of function $h(x)$ invariant. 
\begin{definition}[Control barrier function (CBF)]{
\label{definition: CBF definition}
% \textbf{Definition 3 Control barrier function - CBF}
Given a control-affine system $\dot x=f(x)+g(x)u$, the set $\mathcal{C}$ , with $\frac{\partial h}{\partial \state}(\state) \neq 0$ for all $\state \in \partial \mathcal{C}$, the function $h$ is called the control barrier function (CBF) defined on the set $\mathcal{X}$, if there exists an extended \textit{class}-$\mathcal{K}$ function $\kappa$ such that for all $\state \in \mathcal{X}$:
\begin{equation}
\begin{aligned}
    \underset{ u \in \mathbb{U}}{\text{sup}}\! \left[\underbrace{\mathcal{L}_{f} h(\state) + \mathcal{L}_g h(\state)u} \iffalse+ \frac{\partial h}{\partial t}\fi_{\dot{h}\left(\state, u\right)} \! + \kappa\left(h(\state)\right)\right] \! \geq \! 0,
\end{aligned}
\end{equation}
where $\mathcal{L}_{f} h(\state) = \frac{\partial h}{\partial \state}f(\state)$ and $\mathcal{L}_{g} h(\state)= \frac{\partial h}{\partial \state}g(\state)$ are the Lie derivatives.}
\end{definition}
By \cite{Ames_2017} and \cite{ames2019control}, we have that any Lipschitz continuous control law $\mu(\state)$ satisfying the inequality: $\dot{h} + \kappa( h )\geq 0$ ensures safety of $\mathcal{C}$ if $x(0)\in \mathcal{C}$, and asymptotic convergence to $\mathcal{C}$ if $x(0)$ is outside of $\mathcal{C}$. In case of exponential CBFs, $\kappa(h)=\alpha h, \alpha>0$.
\newline
CBFs are typically regarded as \textit{safety filters} which take the desired input (reference controller input) $u_{ref}(\state,t)$ and modify this input in a minimal way: 

\begin{equation}
\begin{aligned}
\label{eqn: CBF QP}
u^{*}(x,t) &= \min_{u \in \mathbb{U} \subseteq \mathbb{R}^m} \norm{u - u_{ref}(x,t)}^2\\
\quad & \textrm{s.t. } \mathcal{L}_f h(x) + \mathcal{L}_g h(x)u + 
 \kappa \left(h(x)\right) \geq 0.
\end{aligned}
\end{equation}
% \textbf{Solution to Quadratic Program}\\
% Let $h$ be a control barrier function for the control system \ref{eqn: affine control system} and assume that $\mathbb{U} = \mathbb{R}^{m}$.
This is called the Control Barrier Function based Quadratic Program (CBF-QP).
\subsection{Higher order CBFs}
In the above definition, it is assumed that the first- derivative of CBF depends on the control inputs. When the control inputs are a part of higher derivatives, High-Order Control Barrier Functions (HO-CBFs)\cite{DBLP:journals/corr/abs-1903-04706, 9516971} is used.
\begin{definition}[Degree of a function]
    A smooth function $h: \mathbb{R}^n \rightarrow \mathbb{R}$ has relative degree $r \in \mathbb{N}$ at $x \in \mathbb{R}^n$ if :
    \begin{itemize}
        \item $L_gL^i_fh(x) =0 , \forall i \in \{ 0, \cdots,r-2\}$\\
        \item $L_gL^{r-1}_fh(x) \neq 0$
    \end{itemize}
\end{definition}
\begin{definition}[HO- CBFs]
    Let $\mathcal{C} $ be the super-level set of function $h$. Assuming that $h$ has relative degree $r \geq 2$, define a sequence of functions:
    \begin{align*}
        \psi_{i+1}(x(t))=\dot{\psi_i(x(t))}+\kappa_i(\psi_i(x(t)) ,\forall i \in \{ 0,1,\cdots,r-2\}
    \end{align*}
    where $\kappa_i \in \mathcal{K}_e $(class $\mathcal{K}$ functions), $\psi_0(x(t))=h(x(t))$. This produces a collection of sets:
    \begin{align*}
        \mathcal{C}_i:=\{x: x\in \mathcal{X}, \psi_i(x) \geq 0\}, \forall i \in \{0,1,\cdots, r-1\}
    \end{align*}
    with $\mathcal{C}_0=\mathcal{C}$.
    \newline
    The controller $\mu : \mathbb{R}^n \rightarrow \mathbb{R}^m$ satisfying the condition $\psi_r(x(t))=\dot{\psi_{r-1}(x(t))}+\kappa_{r-1}(\psi_{r-1}(x(t)) \geq 0$ renders the set $\mathcal{S}$, defined by :
    \begin{align*}
        S:= \cap_{i=0}^{r-1}\mathcal{C}_i \subseteq \mathcal{C}
    \end{align*}
    invariant.
\end{definition}
\begin{definition} [Higher order exponential CBFs]
    Ex-HO-CBFs from \cite{7524935} has $\kappa(x)=\alpha x$, $\alpha>0$. HO-CBF condition becomes:
    \begin{align*}
        (\frac{d}{dt}+a_1) \circ(\frac{d}{dt}+a_2)\circ \cdots\circ (\frac{d}{dt}+a_r)h(x(t)) \geq 0
    \end{align*}
    $a_i >0 ,\forall i \in \{1,2, \cdots r\}$
\end{definition}
\subsection{Prescribed-time convergence in first order CBFs}
The controller $\mu$ satisfying the condition:
\begin{align*}
    L_f(h(x(t))+L_g(h(x(t))\mu(x(t))+ \kappa(h(x(t)) \geq 0
\end{align*} 
renders the super level set of $h$ invariant. For prescribed time convergence within $T$, $\kappa$ is taken as the following:
\begin{align*}
    k(h(x(t)))=\frac{|h(x(0))|^{1-p}}{(1-p)T}|h(x(t))|^{p}sign(h(x(t)))
\end{align*}
where $\alpha >0, p \in (0,1)$, from \cite{9482751} guarantees the convergence of the system to $\mathcal{C}$ within $T$

\section{Repeated-roots based Pxt-HO-CBF-CLF}

\label{section: Method}
\subsection{Motivation}
The first order prescribed time converging CBF when $h(0) < 0$ is given by :
\begin{align*}
    \Dot{h}=\alpha(-h)^p \quad \alpha \geq 0, p\in (0,1)
\end{align*}

On solving the above differential equation with $\alpha$ and $p$ defined above till $T \text{ such that }h(T)=0$, 
\begin{align*}
   h(t)=h(0)\left(\frac{T-t}{T}\right)^{\frac{1}{1-p}},
\qquad 0\leq t\leq T
\end{align*}

\par The solution to this differential equation gives a function in $t$ and the constants $\alpha$ are chosen such that the only root is at $T$. 
\par In case of order $n$ exponential CBF,
\begin{align*}
    (\frac{d}{dt}+a_1)\circ(\frac{d}{dt}+ a_2)\circ \cdots(\frac{d}{dt}+a_n)h(t) =0  
\end{align*}
\par s.t $a_k \geq 0, \forall k \in \{ 1,2,\cdots ,n\}$.If $a_1=\cdots =a_n=a$, then the solution to the differential equation is 
\begin{align*}
    h(t)=(c_0+c_1t+c_2t^2 +\cdots +c_{n-1}t^{n-1})e^{-at}
\end{align*}
\par with the constants $c_k$ that satisfy the initial conditions $h(0),h^1(0),... , h^{n-1}(0)$. It can be seen that if the polynomial part in $h(t)$ has $T$ as the highest root and $c_{n-1} >0$, then $h(t) \geq 0 , \forall t \in [T,\infty)$.
\newline
The rest of this section is divided into 2 parts, first part provides the method to calculate $a$ that ensures that $T$ is the largest root for this polynomial and the conditions for guarantees for it's existence. Second part shows that this results hold true in case of inequality as well i.e $(\frac{d}{dt}+a_1)\circ \cdots \circ(\frac{d}{dt}+a_n)h(t) \geq 0$ and solving for second order CBFs explicitly.
\subsection{Finding $a$}
\begin{lemma}
    The constants $c_k$ are polynomial functions of $a$ with the highest power being $k$ and they are only functions of $h(0),h^1(0),\cdots, h^k(0)$ \label{Lemma: constants}
\end{lemma} 
\begin{proof}
\begin{align*}
    h(t)= \underbrace{(\sum_{k=0}^{n-1}c_kt^k)}_{p(t)}\underbrace{e^{-at}}_{q(t)}
\end{align*}
Any $m^{th}$ derivative of a function in the form $f(x)=A(x)B(x)$ is given by:
\begin{align}
    \frac{d^mf(x)}{dx^m}=\sum_{n=0}^m\binom{m}{n}\frac{d^nA(x)}{dx^n}\frac{d^{m-n}B(x)}{dx^{m-n}} \label{eq:diff}
\end{align}
where $\binom{n}{r}$ denotes ${}^nC_r$.
Using \eqref{eq:diff}, $i^{th}$ derivative of $h(t)$ is given by:
\begin{align}
    \frac{d^ih(t)}{dt^i}=\sum_{k=0}^{i} \binom{i}{k}\frac{d^kp(t)}{dt^k}\frac{d^{i-k}q(t)}{dt^{i-k}}
\end{align}
Using
\begin{align*}
    \frac{d^ip(t)}{dt^i}=&\sum_{k=i}^{n-1}\frac{k!}{(k-i)!}c_kt^{k-i}, \frac{d^ip(0)}{dt^i}=i!c_i\\
    \frac{d^iq(t)}{dt^i}=&(-a)^ie^{-at}, \frac{d^iq(0)}{dt^i}=(-a)^i
\end{align*}
in the equation from $h^i(0)$
\begin{align}
    h^i(0)=\frac{d^ih(0)}{dt^i}=\sum_{k=0}^i\binom{i}{k}(k!c_k)(-a)^{i-k} \label{eq:final h(0)}
\end{align}
From \eqref{eq:final h(0)}, finding $c_k$ is same as solving the system of linear equation:
\begin{equation*}
\resizebox{\linewidth}{!}{$
\begin{bmatrix}
    h(0)\\
    h^{(1)}(0)\\
    \vdots\\
    h^{(n-1)}(0)
\end{bmatrix}
=
\underbrace{\begin{bmatrix}
    1 & 0 & \cdots & 0\\
    -a & 1 & \cdots & 0\\
    \vdots & \vdots & \ddots & \vdots\\
    (-a)^{n-1} & (n-1)(-a)^{n-2} & \cdots & (n-1)!
\end{bmatrix}}_{A}
\begin{bmatrix}
    c_0\\
    c_1\\
    \vdots\\
    c_{n-1}
\end{bmatrix}
$}
\end{equation*}
% $A$ is lower triangular, so the system can be solved by forward substitution. In particular, when determining $c_{k+1}$, the coefficients $c_0,\ldots,c_k$ have already been determined. Since the exponent of $a$ in the $(i,j)$-th entry of $A$ is $i-j$ for $i\geq j$,
% $c_{k+1}=b\left(ac_k,a^2c_{k+1},\ldots,a^kc_0\right).$

%Suppose that the highest power of $a$ appearing in $c_k$ is $a^k$. Then the highest power of $a$ appearing in each term $a^j c_{k+1-j}$ is
%$j+(k+1-j)=k+1.$
%Hence, the highest power of $a$ in $c_{k+1}$ is $a^{k+1}$. Since
%$c_0=h(0),$
%the statement follows by induction.

% $A^{-1}$ exists because all diagonal entries of $A$ are non-zero constants and is lower triangular. Therefore,
%$c_k=\sum_{i=0}^kb_ih^i(0)$

$A$ is lower triangular matrix, so $A^{-1}$ is also lower triangular, also the inverse exists as all the diagonal elements of $A$ are non-zero constants. The element at $i,j$ of $A$ is $[A]_{i,j}=\binom{i}{j}j!(-a)^{i-j} , i \geq j \quad \text{else } 0$. Finding the elements of $A^{-1}$ :
\newline
Let $B=A^{-1}$ ,
\begin{align*}
    [AB]_{i,j}=\sum_{k=1}^n[A]_{i,k}[B]_{k,j}
\end{align*}
Since $A, B$ are lower triangular, $[A]_{i,j}=[B]_{i,j}=0 $ if $j>i$. $[AB]_{i,j}=\sum_{k=j}^i[A]_{i,k}[B]_{k,j}, i \geq j$ else $0$.
On using binomial theorem $ f_{i,j}(x)=(1+x)^{i-j}=1 $ if $i=j$ and $f_{i,j}(-1)=0 \forall i>j$, $B=A^{-1},$ and  $\binom{i}{k}\binom{k}{j}=\binom{i}{j}\binom{i-j}{k-j}$, 
\begin{align*}
    [A^{-1}]_{i,j}&=\frac{1}{i!}\binom{i}{j}(a)^{i-j}, i \geq j \quad \text{else } 0\\
    c_k=&\sum_{i=0}^k\frac{1}{k!}\binom{k}{i}a^{k-i}h^i(0)
\end{align*}
From this, the highest power of $a$ in $c_k$ is $k$ and it only a function of $h(0),h^1(0),\cdots, h^k(0)$
\end{proof}
The $h(t)$ is given by:
\begin{align}
    h_a(t)=\underbrace{(\sum_{k=0}^{n-1}(\sum_{i=0}^k\frac{1}{k!}\binom{k}{i}a^{k-i}h^i(0))t^k)}_{p_a(t)}e^{-at}
\end{align}
Since $ e^{-at} \geq 0 \quad \forall t \geq 0$, $h_a(t) \geq0  \iff p_a(t) \geq 0$. So we need $p_a(t) \geq 0 \quad \forall t \in [T,\infty)$.
\newline
\textbf{If $max(R(p_a(t))) \leq T$, where $R(f(t))$ gives the roots of $f(t)$, and coefficient of $p_a(t)$ corresponding to the highest power of $t$ i.e $n-1$ is greater than 0, then $p_a(t) \geq 0 \quad\forall t \in [T,\infty)$.}
\newline
Let $b_a(t)=p_a(t+T)$, so we need $b_a(t) \geq 0 \quad \forall t \in [0,\infty)$, so the largest root to be less than 0, $max(R(b_a(t)))\leq 0$. To find $a$ that satisfies the above condition, we use Sturm's Theorem.
\subsection{Sturm's theorem}
 For a polynomial $p(x)$, Sturm's sequence is defined by $p_0(x)=p(x)$, $ p_1(x)=p'(x)=\frac{dp(x)}{dx}$, $p_{i+1}(x)=-rem\{p_{i-1}(x),p_{i}(x)\}$, till the remainder is constant. Then the number of distinct real roots in the interval $(a,b], a < b$ is given by $V(a)-V(b)$, where $V(a)$ is the number of sign changes in the sequence $p_i(a)$ [Refer \cite{raghavan2016sturm} for further details. This approach has also been employed in related polynomial root-location problems (see \cite{FORDE2004273}).]
\newline
For example: if $p(x)= x^2+3x+2$, 
$p_0(x)= x^2+3x+2, p_1(x)= 2x+3, p_2(x)=\frac{1}{4}$. The roots of $p(x)$ are -1 and -2. 
\begin{itemize}
    \item In the interval $(0,\infty)$: $V(0)=[+, +,+]=0$,$V(\infty)=[+,+,+]=0$,  so $ V(0)-V(\infty)=0$\\
    \item In the interval $(-1.4, \infty)$: $V(-1.4)=[-,+,+]=1$, $V(\infty)=[+,+,+]=0$, $V(-1.4)-V(\infty)=1$\\
    \item In the interval $(-2.5,\infty)$: $V(-2.5)=[+,-,+]=2$, $V(\infty)=[+,+,+]=0$, $V(-2.5)-V(\infty)=2$
\end{itemize}
Using Sturm's theorem between the interval $(0, \infty)$ for the function $b_a(t)$- $V(0)-V(\infty)=0$, the conditions on $a$ can be found and $\sum_{i=0}^{n-1}\frac{1}{(n-1)!}\binom{n-1}{i}a^{n-1-i}h^i(0) >0$ 
\subsection{Guarantees on the existence of $a$ }
\begin{align*}
    b_a(t)=& h(0)+\sum_{k=1}^{n-1}(\sum_{i=0}^k\frac{1}{k!}\binom{k}{i}a^{k-i}h^i(0))(t+T)^k\\
\end{align*}
From this clearly, when $|a| >> |h(0)|$, for $b_a(t)=0$, the roots have to be in the neighbourhood of $T: |t+T|< \epsilon \forall t \in R(b_a(t))$, now from this, real roots are $t=-T \pm \epsilon$, for large $a$, $\epsilon\approx\frac{1+(n-1)!}{|a|}<T$, $t<0$ ( $\epsilon\approx\frac{1+(n-1)!}{|a|}$ is found using Cauchy's upper limit for absolute value of roots of  and using $a(t+T)$ as the variable, assuming $c_{n-1} \neq 0$ which is true), such $a$ give all negative roots, so this satisfies the Sturm's condition, so $a$ must exist for all $b_a(t)$.
\begin{remark}
    Note that there is no guarantee on the existence of $a$ that satisfies $ \sum_{i=0}^{n-1}\frac{1}{(n-1)!}\binom{n-1}{i}a^{n-1-i}h^i(0) >0$, when $n$ is odd (in these cases, the proposed method cannot be used),when $n$ is even, $a$ exists that satisfies the above condition. Even if exists, there is no guarantee that  $a>0$. In cases where there is no $a>0$, the one that is closest to $0$ is used and once the safe set is reached any $a_1,a_2$ according to Ex-HOCBF chosen, this still maintains the forward invariance of the safe-set. It is also to be noted that there are switching effects in the existing methods as well. 
\end{remark}
\subsection{Extending the results to inequality for QP formulation} 
\begin{theorem}[Inequality Guarantee via Green's Function]
Consider the linear differential operator $L = \left(\frac{d}{dt} + a\right)^n$. Let $u(t)$ be the solution to the initial value problem:
\begin{equation}
    L[u] = 0,  u^{(k)}(0) = u_{0,k} , k=0,\dots,n-1.
\end{equation}

Let $v(t)$ be any $n$-times differentiable function having the same initial conditions, $v^{(k)}(0) = u_{0,k}$, and the inequality:
\[
L[v](t) \geq 0 \quad \text{for all } t \geq 0.
\]
Then, $v(t) \geq u(t)$ for all $t \geq 0$.
\end{theorem}

\begin{proof}
Let $w(t) = v(t) - u(t)$. By linearity of the operator $L$, we have:
\begin{equation}
    L[w] = L[v] - 0 \geq  0.
\end{equation}

The initial conditions for $w$ are zero: $w^{(k)}(0) = 0$ for all $k = 0, \dots, n-1$, as both $v(t)$ and $u(t)$ have same initial conditions.

The solution to the initial value problem $L[w] = g(t)$ with zero initial conditions can be expressed using the Green's function $G(t,\tau)$:
\[
w(t) = \int_0^t G(t-\tau) g(\tau) d\tau,
\]
where $G(t-\tau)$ is the impulse response of the operator $L$ with zero initial conditions.

For the operator $L = \left(\frac{d}{dt} + a\right)^n$, the Green's function is given by:
\[
G(t-\tau) = \frac{(t-\tau)^{n-1}}{(n-1)!} e^{-a(t-\tau)} \quad \text{for } t \geq \tau,
\]
and $G(t-\tau) = 0$ for $t < \tau$.

Therefore, $G(t-\tau) \geq 0$ for all $t, \tau$.

Using $g(\tau) = L[w](\tau) \geq 0$ into the convolution integral:
\[
w(t) = \int_0^t \underbrace{G(t-\tau)}_{\geq 0} \underbrace{g(\tau)}_{\geq 0} d\tau \geq 0.
\]
Hence, $v(t) = u(t) + w(t) \geq u(t)$ for all $t \geq 0$.
\end{proof}
\begin{corollary}
For the prescribed-time CBF formulation, if the CBF-QP enforces:
\[
\left(\frac{d}{dt} + a\right)^n h(x(t)) \geq 0
\]
 $ a $ determined by repeated-roots method, and the equality case solution ,$h_{eq}(t)$ ,satisfies $h_{eq}(T) \geq 0$ and $h_{eq}(t) \geq 0$ for $t \geq T$, then the CBF-QP solution $h(t)$ will satisfy $h(t) \geq h_{eq}(t)$ and consequently $h(t) \geq 0$ for all $t \geq T$.
\end{corollary}
\subsection{Second order CBF/CLF}
Most of the systems in real-life are second order CBFs, for example goal reaching $h(x)=||x-x_g||^2_2$ in unicycle model, dubins car model are all second order. In this we get the explicit conditions for second order CBFs.
\begin{theorem}[Second-order-Prescribed-Time CLBF]\label{thm:second_hocbf}
Let \( h(x) \in C^2 \) be a twice-differentiable control barrier function with relative degree two, and let the initial conditions \( h(0) < 0 \), \( h^1(0) \in \mathbb{R} \) be given. For a desired convergence time \( T > 0 \), define
\begin{equation}
    a = \frac{-h(0) - h^1(0) T}{h(0) T}.
\end{equation}
Then enforcing the differential constraint
\begin{equation}
    \left( \frac{d}{dt} + a \right)^2 h(x(t)) \geq 0
\end{equation}
ensures that \( h(t) \geq 0 \) at time \( t = T \), and \( h(t) > 0 \) for all \( t > T \) .
    
\end{theorem}
\begin{proof}
    On solving the differential equation:
    \begin{align*}
        h(t)=(a_0+a_1t)e^{-at}
    \end{align*}
    Using Lemma \ref{Lemma: constants}, $a_0=h(0), a_1=ah(0)+h^1(0)$.
    \newline
    $p(t)=h(0)+(ah(0)+h^1(0))t$. On applying Sturm's theorem to $p(t)$ and getting the conditions that roots in $(T,\infty)$ is 0 :
    $p_0(t)=h(0)+(ah(0)+h^1(0))t, p_1(t)= ah(0)+h^1(0)$
    Since $p(t) >0 $ as $t \rightarrow \infty$, $a_1=ah(0)+h^1(0) >0$.
    $V(T)=[sign(p_0(T)), sign(ah(0)+h^1(0))]$,
    $V(\infty)=[sign(a_1),sign(a_1)]=0$ since $V(T)-V(\infty)=0$,$V(T)=0$, so $p_0(T) \geq0$ as $a_1\geq0$ as they have to be of the same sign. The two conditions are:
    \begin{align*}
        h(0)+(ah(0)+h^1(0))T \geq0 \implies a \leq \frac{-h(0)-h^1(0)T}{h(0)T}\\
        ah(0)+h^1(0) >0 \implies a <\frac{-h^1(0)}{h(0)} 
    \end{align*}
    Since $T>0$,
    \begin{align*}
        a \leq &\frac{-h(0)-h^1(0)T}{h(0)T} < \frac{-h^1(0)}{h(0)} \\
        a = &\frac{-h(0)-h^1(0)T}{h(0)T}
    \end{align*}
    So taking $a=\frac{-h(0)-h^1(0)T}{h(0)T}$ gives the one that is closest to $0$ when $\frac{-h(0)-h^1(0)T}{h(0)T} <0$ and if $\frac{-h(0)-h^1(0)T}{h(0)T}>0$ then gives a positive value which is what is needed.
\end{proof}
\begin{corollary}
    If $\frac{h^1(0)}{h(0)} \leq -\frac{1}{T}$, $\exists \quad  a \geq 0$. If $h^1(0) \geq \frac{|h(0)|}{T}, \exists \quad a \geq 0$.
\end{corollary}

\section{Simulations}
\label{section: Simulation}

We evaluate the proposed RR-Pxt-HOCBF on two systems: an inverted pendulum and an acceleration-controlled unicycle. The method is compared against STL-HOCLBF \cite{xiao2021highordercontrollyapunovbarrier} and the Ex-HOCBF (denoted HOCBF) from \cite{https://doi.org/10.48550/arxiv.1903.04706}. The QP optimisation was to minimize control effort $||u||^2_2$

\subsection{Inverted Pendulum}

The state is $x = [\theta,\dot\theta]^T$ and the dynamics are
\begin{equation}
\dot{x}=
\underbrace{\begin{bmatrix}
\dot{\theta} \\
\frac{-mgl\sin(\theta)-b\dot{\theta}}{ml^2}
\end{bmatrix}}_{f(x)}
+
\underbrace{\begin{bmatrix}
0\\
1
\end{bmatrix}}_{g(x)}\tau.
\end{equation}

Parameters: $l=10\,$m, $m=1\,$kg, $b=0$, $g=10m/s^2$.  
Safety barrier function: $h(x)=0.0025-\theta^2$.  
Discrete update: $x_{t+1}=x_t+(f(x_t)+g(x_t)\tau)\Delta t$, with $\Delta t=0.001s$.  
Initial condition: $x_0=[\pi,0]^T$, prescribed time $T=10\,$s.

Figure~\ref{fig:Pendulum} shows that HOCBF converges at $t\approx 10.95$\,s (violates prescribed time), STL-HOCLBF at $\approx 8.6$\,s, and RR-Pxt-HOCBF at $\approx 6.25$\,s.

\begin{figure}[h]
\centering
\includegraphics[width=0.7\linewidth]{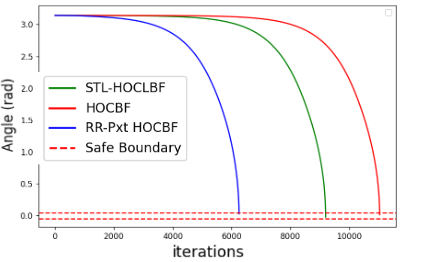}
\caption{Inverted Pendulum: $\theta(t)$}
\label{fig:Pendulum}
\end{figure}

\begin{figure}[h]
\centering
\includegraphics[width=0.7\linewidth]{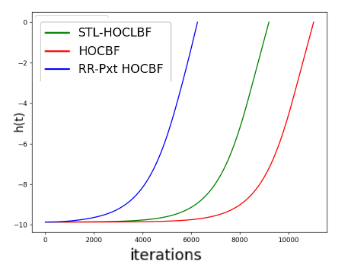}
\caption{Inverted Pendulum: $h(t)$}
\label{fig:pendht}
\end{figure}

\subsection{Acceleration-Controlled Unicycle}

The state is $x=[x_p,y_p,\theta,v]^T$ with dynamics
\begin{equation}
\dot{x}=
\underbrace{\begin{bmatrix}
v\cos\theta\\
v\sin\theta\\
0 \\
0 
\end{bmatrix}}_{f(x)}
+
\underbrace{\begin{bmatrix}
0 & 0\\
0 & 0\\
1 & 0\\
0 & 1
\end{bmatrix}}_{g(x)}
\begin{bmatrix}
\omega\\ a
\end{bmatrix}.
\end{equation}

Safety barrier function:
\[
h(x)=4-(x_p-45)^2-(y_p-45)^2.
\]
Discrete update with $\Delta t=0.001s$.  
Initial condition: $x_0=[0,0,0,0]^T$.  
Prescribed time: $T=25\,$s.

Figure~\ref{fig:Unicycle} shows that HOCBF fails to reach the safe set even after $2T$, while STL-HOCLBF and RR-Pxt-HOCBF converge close to prescribed time ($25.6$\,s and $25.08$\,respectively- $25.08$ is clearly due to discretization).

\begin{figure}[h]
\centering
\includegraphics[width=0.7\linewidth]{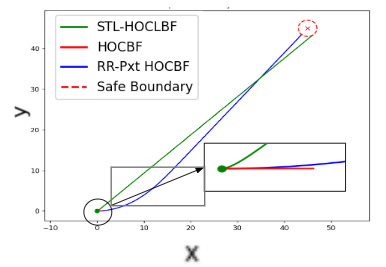}
\caption{Unicycle Trajectory $(x,y)$}
\label{fig:Unicycle}
\end{figure}
\begin{figure}[h]
\centering
\includegraphics[width=0.7\linewidth, height=0.4\linewidth]{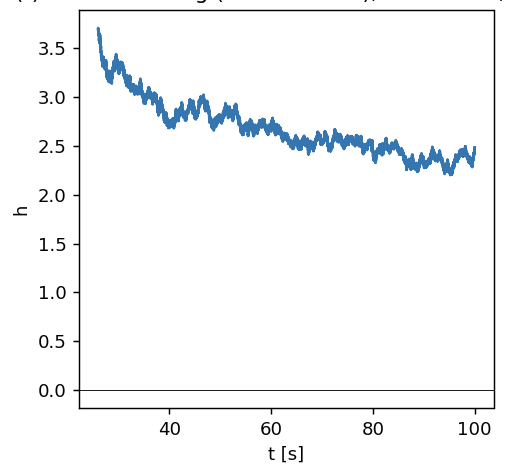}
\caption{$h(t)$ after switching ($t=26s$) from $a=-0.04$ to $a_1=1000, a_2=1$ (chosen by Ex-HOCBF) maintains the system within the safe set - switching doesn't affect invariance of safe-set after $T$. }
\label{fig:uniht}
\end{figure}
\subsection{Comparison With Existing Methods}

STL-HOCLBF requires $2n-1$ parameters and they have to satisfy non-linear equation \ref{non}, and $x(t_i)$ is required which is not known before hand, as it is a function of control inputs. At each iteration in STL-HOCLBF   $n-1$ intermediate values of $\psi_i(x)$ should be found increasing the computation time per iteration significantly. This can be observed in simulation; STL-HOCBF computation time per iteration is higher than the proposed method. HOCBF and RR-Pxt HOCBF have similar time as RR-Pxt HOCBF is a special case of HOCBF.
\begin{equation} \label{non}
T \geq \sum_{i=1}^{n} \frac{\left( \psi_{i-1} \left( x(t_i) \right) \right)^{1 - q_i}}{p_i (1 - q_i)},
\end{equation}

Table~\ref{tab:comparison}  contains the comparison results in unicycle model and initial condition is $[0,0,0,0]^T$ and $T=25s$ and Table \ref{tab:comparison_pend} for inverted pendulum model when initial condition is $[\pi,0]^T$ and $T=10s$. 

\begin{table}[h]
\caption{Comparison (Unicycle)}
\label{tab:comparison}
\centering
\scriptsize
\setlength{\tabcolsep}{4pt}
\resizebox{\columnwidth}{!}{%
\begin{tabular}{|l|c|c|c|}
\hline
\textbf{Metric} & \textbf{HOCBF} & \textbf{STL-HOCBF} & \textbf{RR-Pxt-HOCBF} \\
\hline
Computation Time & 4.2 ms* & 7.62 ms  & 4 ms  \\
Goal reached & No & Yes & Yes \\
Parameters & 2 & 3 & 0 \\
\hline
\end{tabular}%
}
\end{table}

\begin{table}[h]
\caption{Comparison (Pendulum)}
\label{tab:comparison_pend}
\centering
\scriptsize
\setlength{\tabcolsep}{4pt}
\resizebox{\columnwidth}{!}{%
\begin{tabular}{|l|c|c|c|}
\hline
\textbf{Metric} & \textbf{HOCBF} & \textbf{STL-HOCBF} & \textbf{RR-Pxt-HOCBF} \\
\hline
Computation Time & 4.6 ms &  5 ms & 4.5 ms \\
Goal reached & No & Yes & Yes \\
Parameters & 2 & 3 & 0 \\
\hline
\end{tabular}%
}
\end{table}
(*HOCBF did not converge within $T$; simulation continued, computation time is time per iteration)

\subsection{Studying the effect of Discretization on RR-Pxt-HOCBF}
As in $\Delta t \to 0$, the discrete approximation matches the behavior in continuous-time better, reducing the deviation in convergence time. Table \ref{tab:discrete} shows the effect of discretization, this explain why in Unicycle model, the method doesn't converge within prescribed time- $10 s$, but as $\Delta t \rightarrow 0$, $t \rightarrow 10s$ where $t$ is the time when system converges.
\begin{table}[h]
\caption{Effect of Discretization}
\centering
\small
\begin{tabular}{|l|c|c|c|}
\hline
\textbf{Model}  & $\Delta t=0.0001s$ & $\Delta t=0.001s$ & $\Delta t=0.01s$\\
\hline
Pendulum & 6.25 s & 6.255 s & 6.29 s \\
Unicycle & 10.02 s & 10.08 s & 10.36 s \\
\hline
\end{tabular}
\label{tab:discrete}
\end{table}

\begin{remark}
Initial conditions: $[\pi, 0]^T$ for the pendulum and $[0,0,0,0]^T$ for the unicycle, with $T=10$\,s for both.
\end{remark}

\section{Conclusions}
\label{section: Conclusions}
% In this section, we will discuss about the limitations of our method and the future works.
We have proposed a method, RR-Pxt-HOCBF-CLF, that achieves prescribed-time convergence with reduced computational complexity and fewer design parameters, while allowing the convergence time to be specified by the user. The proposed approach demonstrates improved performance compared to existing methods,even in cases where $a<0$, switching to HOCBF doesn't cause any problem. However, one limitation is that the actual convergence time slightly exceeds the prescribed time due to discretization effects, and the control inputs tend to be larger than those in previous approaches. Future work will focus on quantifying the impact of discretization and extending the framework to handle systems with input and state constraints but the proposed method is clearly better than the existing methods when $a$ exists satisfying the condition $c_{n-1} >0$.

\label{section: References}
\bibliographystyle{IEEEtran}
\bibliography{references.bib}

\end{document}